\documentclass[aps,prx,reprint,amsmath,amssymb,longbibliography,floatfix]{revtex4-2}

\usepackage[T1]{fontenc}
\usepackage[utf8]{inputenc}
\usepackage{lmodern}
\usepackage{amsthm,mathtools,bm}
\usepackage{booktabs,graphicx,microtype,array}
\usepackage[hidelinks]{hyperref}
\usepackage[nameinlink,noabbrev]{cleveref}

\hypersetup{
  pdftitle={Readout-Rank Laws for Isotropic Quantum Tangents},
  pdfauthor={Marwan Ait Haddou}
}

\newtheorem{theorem}{Theorem}
\newtheorem{proposition}[theorem]{Proposition}
\newtheorem{corollary}[theorem]{Corollary}

\newcommand{\FQ}{F_Q}
\newcommand{\Ffull}{F_{\mathrm{full}}}
\newcommand{\Iacc}{\mathcal I}
\newcommand{\A}{\mathcal A}
\newcommand{\E}{\mathbb E}
\newcommand{\Var}{\operatorname{Var}}
\newcommand{\Prob}{\mathbb P}
\newcommand{\rank}{\operatorname{rank}}
\newcommand{\supp}{\operatorname{supp}}

\begin{document}

\title{Readout-Rank Laws for Isotropic Quantum Tangents}
\author{Marwan Ait Haddou}
\email[Contact author: ]{aithaddou.marwan@outlook.com}
\affiliation{Independent Researcher, Morocco}

\begin{abstract}
Deep parameterized quantum circuits may remain sensitive to a parameter change while the observables retained by a learning model barely respond.  We study this separation for a fixed computational-basis measurement.  For a pure-state tangent, we compare the quantum Fisher information $F_Q$, the Fisher information $F_{\rm full}$ in the complete bitstring distribution, and the largest variance-normalized response $\mathcal I_{\mathcal A}$ available to a diagonal readout space $\mathcal A$.  If the joint state--tangent frame is Haar random, we prove that the two successive information fractions are independent Beta variables whose means are $1/2$ and $r/(2^n-1)$, where $r$ is the centered dimension of the readout.  Consequently, even the joint span of all computational-basis Pauli strings through any fixed weight $k$ retain only $O(n^k2^{-n})$ of the full-record information.  Exact-statevector experiments across six circuit families show increasing finite-size agreement with this hierarchy in five nonconserving ensembles as the circuit depth grows.  A number-conserving family departs strongly from the isotropic prediction even after correcting the support and readout rank, showing that rank alone is insufficient without tangent isotropy.
\end{abstract}

\maketitle

\section{Introduction}
\label{sec:introduction}

Variational quantum algorithms couple a parameterized state
preparation to a classical optimizer through measured
expectation values or outcome statistics
\cite{peruzzo2014,mcclean2016}. The quantum state can
therefore respond strongly to a circuit parameter while the
quantities supplied to the optimizer respond only weakly. The
first statement concerns motion in state space; the second
concerns how much of that motion survives the chosen
measurement and the subsequent classical compression. Keeping
these two questions separate is the starting point of this work.

The distinction is particularly relevant in circuit-based
quantum classifiers, whose predictions are commonly constructed
from measured observables
\cite{havlicek2019,schuld2020,perezsalinas2020}. A circuit may
be sampled in the computational basis while the learning model
retains only a restricted set of quantities, such as the
single-qubit expectations $\langle Z_i\rangle$, a collection of
low-weight Pauli correlators, or a marginal distribution. These
features are computed from the complete bitstring record, but
they need not retain the same local statistical information.
The data-processing inequality orders the resulting
descriptions \cite{zamir1998}, but it does not determine the
size of the information loss.

More generally, global state distinguishability does not
guarantee that a restricted set of measured features exposes the
relevant signal. Low-weight observables may retain only a small
part of the information contained in the complete measurement
record. This creates a readout bottleneck that is logically
distinct from weak state sensitivity: the state and the full
outcome distribution may both respond appreciably even when the
features supplied to a learning model barely change.

Restrictions on accessible measurements already play a central
role in quantum estimation. Optimization over all measurements
leads to the quantum Fisher bound, whereas restrictions on the
allowed observables produce smaller information hierarchies
\cite{braunstein1994,hotta2004,lu2012}. Geometric approaches to
variational circuits likewise use the Fubini--Study or quantum
Fisher metric to characterize state sensitivity and optimization
geometry \cite{stokes2020,larocca2022}. Vanishing cost gradients
have separately been connected to circuit expressibility, cost
locality, noise, entanglement, and the dynamical Lie algebra
\cite{mcclean2018,cerezo2021,holmes2022,wang2021,
marrero2021,ragone2024}. Those results concern a specified cost
or an ensemble of costs. We instead ask a prior question: given
a nonzero state tangent and a fixed measurement record, how much
directional Fisher information can be reached by any observable
in an allowed readout space?

The answer is controlled by projection geometry. For an
$n$-qubit pure state, the physical tangent has two real sectors
of equal dimension. In the local phase frame, one sector changes
computational-basis probabilities, whereas the other changes
only relative phases to first order. A fixed projective
measurement retains the probability-changing sector and is
locally blind to the phase-changing sector. A restricted
diagonal readout then projects the visible probability tangent
once more onto the score subspace generated by the retained
features. When the joint state--tangent frame is isotropic,
these two losses become orthogonal projections whose exact
finite-size distributions depend only on the dimensions of the
retained subspaces.

The underlying finite-dimensional hyperspherical projection law
is classical: the squared norm of the projection of a uniformly
distributed sphere vector onto a fixed subspace follows a Beta
distribution determined by the retained and discarded
dimensions. This law has recently also been used to derive exact
subsystem-probability statistics for Haar-random pure states
\cite{wang2026geometric}. Our contribution is to identify the
same geometric law with two distinct information losses in a
parameterized quantum model: the loss from the physical tangent
to a fixed-basis probability tangent, and the subsequent loss
from the complete probability record to a rank-restricted
diagonal readout.

Randomizing the measurement basis provides a complementary route to
recovering quantum-geometric information. For pure-state models,
Haar-averaging the classical Fisher information matrix (CFIM) over
measurement bases yields one half of the quantum Fisher information
matrix (QFIM), and recent work gives dimension-dependent variance and
concentration bounds for this randomized-basis estimator
\cite{lu_sha_2025_random_cfim}. Related randomized-measurement schemes
use unitary-design rotations to estimate quantum-geometric information
with reduced state-preparation cost in variational imaginary-time
evolution \cite{kolotouros2025_random_measurements}. Our setting is
different: the computational basis is held fixed while the joint
state--tangent frame is isotropic, and a second projection then
compresses the resulting score into a centered rank-$r$ diagonal
readout. Thus the common factor $1/2$ has a related geometric origin,
but the exponential fixed-order suppression studied below is
specifically a fixed-record readout-compression effect.

This identification leads to an exact readout-rank law for Haar
state--tangent frames. We compare the directional quantum Fisher
information $F_Q$, the Fisher information $F_{\mathrm{full}}$
contained in the complete computational-basis distribution, and
the maximum variance-normalized response
$\mathcal I_{\mathcal A}$ accessible to a diagonal readout space
$\mathcal A$. We show that the two successive fractions
\[
\frac{F_{\mathrm{full}}}{F_Q}
\qquad\text{and}\qquad
\frac{\mathcal I_{\mathcal A}}{F_{\mathrm{full}}}
\]
are independent Beta variables with means $1/2$ and $r/N$,
respectively, where $N=2^n-1$ is the centered score-space
dimension and $r$ is the centered readout rank. Consequently,
\[
\mathbb E\!\left[
\frac{\mathcal I_{\mathcal A}}{F_Q}
\right]
=
\frac{r}{2N}.
\]
Thus, once the state--tangent frame is isotropic, the identity
and redundancy of individual observables matter only through
the dimension of the centered subspace that they jointly span.

This dimension count has a direct consequence for low-order
diagonal readouts. The joint space generated by all
computational-basis Pauli strings of weight at most $k$ has
centered rank
\[
r_k=\sum_{j=1}^{k}\binom{n}{j}.
\]
For every fixed $k$, this rank grows only polynomially in $n$,
whereas the complete score space has dimension $2^n-1$.
Accordingly, the expected retained fraction satisfies
\[
\mathbb E\!\left[
\frac{\mathcal I_{\leq k}}{F_{\mathrm{full}}}
\right]
=
\Theta\!\left(n^k2^{-n}\right).
\]
The suppression therefore does not arise because only one local
observable is retained: it persists even when every diagonal
Pauli string through a fixed weight is used jointly.

The exact laws require joint state--tangent isotropy and do not
follow from randomness of the state alone. We therefore do not
assume that finite-depth parameterized circuits are Haar random.
Instead, exact-statevector experiments test the predicted
finite-size hierarchy across six circuit families. Five
nonconserving families show increasing agreement with the
rank-controlled prediction as the circuit depth grows. A
number-conserving family remains far from the isotropic law even
after the measurement support and the actual readout rank are
corrected. This boundary case shows that rank alone is
insufficient outside the isotropic regime: the orientation of
the probability tangent relative to the readout subspace remains
essential.

The paper is organized as follows. Section~2 defines the state,
measurement, and readout information quantities and derives the
score-projection identity. Section~3 establishes the exact Haar
state--tangent laws and states the conditions under which a
circuit produces the required joint frame. Section~4 presents
the Haar controls, the finite-depth circuit crossover, the
aggregate readout-rank test, and the number-conserving boundary
case. Section~5 discusses implications for
quantum-machine-learning gradients and the limitations of the
analysis. Section~6 concludes. Supporting derivations and
numerical checks are given in the appendices.
\section{Projection Geometry of Accessible Readout Information}

\label{sec:setup}

Let $|\psi_{\bm\theta}\rangle\in\mathbb C^D$ be a differentiable pure state, with $D=2^n$.  Fix a unit parameter-space direction $v$, use $\partial_v$ for the derivative along that direction, and evaluate all quantities at the reference parameter $\bm\theta$ unless stated otherwise.  The raw directional tangent is
\begin{equation}
 |\dot\psi_v\rangle
 =\left.\frac{\partial}{\partial\alpha}
 |\psi_{\bm\theta+\alpha v}\rangle\right|_{\alpha=0}.
 \label{eq:raw_tangent}
\end{equation}
The physically relevant tangent is its horizontal component in projective Hilbert space~\cite{provost1980}
\begin{equation}
 |\phi\rangle=(\mathbb I-|\psi\rangle\!\langle\psi|)|\dot\psi_v\rangle,
 \qquad \langle\psi|\phi\rangle=0.
 \label{eq:horizontal}
\end{equation}
For a pure state, the directional quantum Fisher information is~\cite{braunstein1994}
\begin{equation}
 \FQ(v)=4\langle\phi|\phi\rangle.
 \label{eq:qfi}
\end{equation}
The projection in \cref{eq:horizontal} removes a change of global phase, which has no observable effect.  The remaining vector $|\phi\rangle$ is the instantaneous motion of the physical state in projective Hilbert space, and $F_Q(v)$ is four times its squared Fubini--Study speed~\cite{provost1980,stokes2020}.  It therefore provides the measurement-independent reference against which we compare a fixed readout.

Let $z\in\{0,1\}^n$ denote a computational-basis outcome.  We write the corresponding state amplitude as $\psi_z=\sqrt{p_z}e^{i\vartheta_z}$, where $p_z$ is the measurement probability, and resolve the tangent in the local phase frame as
\begin{equation}
 e^{-i\vartheta_z}\phi_z=x_z+i y_z,
 \qquad x,y\in\mathbb R^D.
 \label{eq:xy}
\end{equation}
Orthogonality gives $x\perp\sqrt p$ and $y\perp\sqrt p$.  Hence both live in real subspaces of dimension
\begin{equation}
 N\equiv D-1.
 \label{eq:N}
\end{equation}
The probability tangent under the computational-basis measurement is
\begin{equation}
 q_z\equiv\partial_v p_z
 =2\operatorname{Re}(\psi_z^*\phi_z)
 =2\sqrt{p_z}\,x_z.
 \label{eq:q}
\end{equation}
Assuming $p_z>0$ for every outcome, the classical Fisher information in the complete bitstring distribution is~\cite{fisher1922}
\begin{equation}
 \Ffull(v)=\sum_z\frac{q_z^2}{p_z}=4\|x\|^2,
 \qquad
 \FQ(v)=4(\|x\|^2+\|y\|^2).
 \label{eq:amplitude_phase}
\end{equation}
Thus a fixed basis sees the real vector $x$, which changes outcome probabilities, but not the vector $y$, which changes only relative phases to first order.  The difference $\FQ-\Ffull=4\|y\|^2$ is a loss at the measurement stage.  A second and conceptually different loss occurs when the complete probability record is compressed to a restricted family of observables.  Figure~\ref{fig:geometry} separates these two projections.

\begin{figure*}[t]
 \centering
 \includegraphics[width=0.98\textwidth]{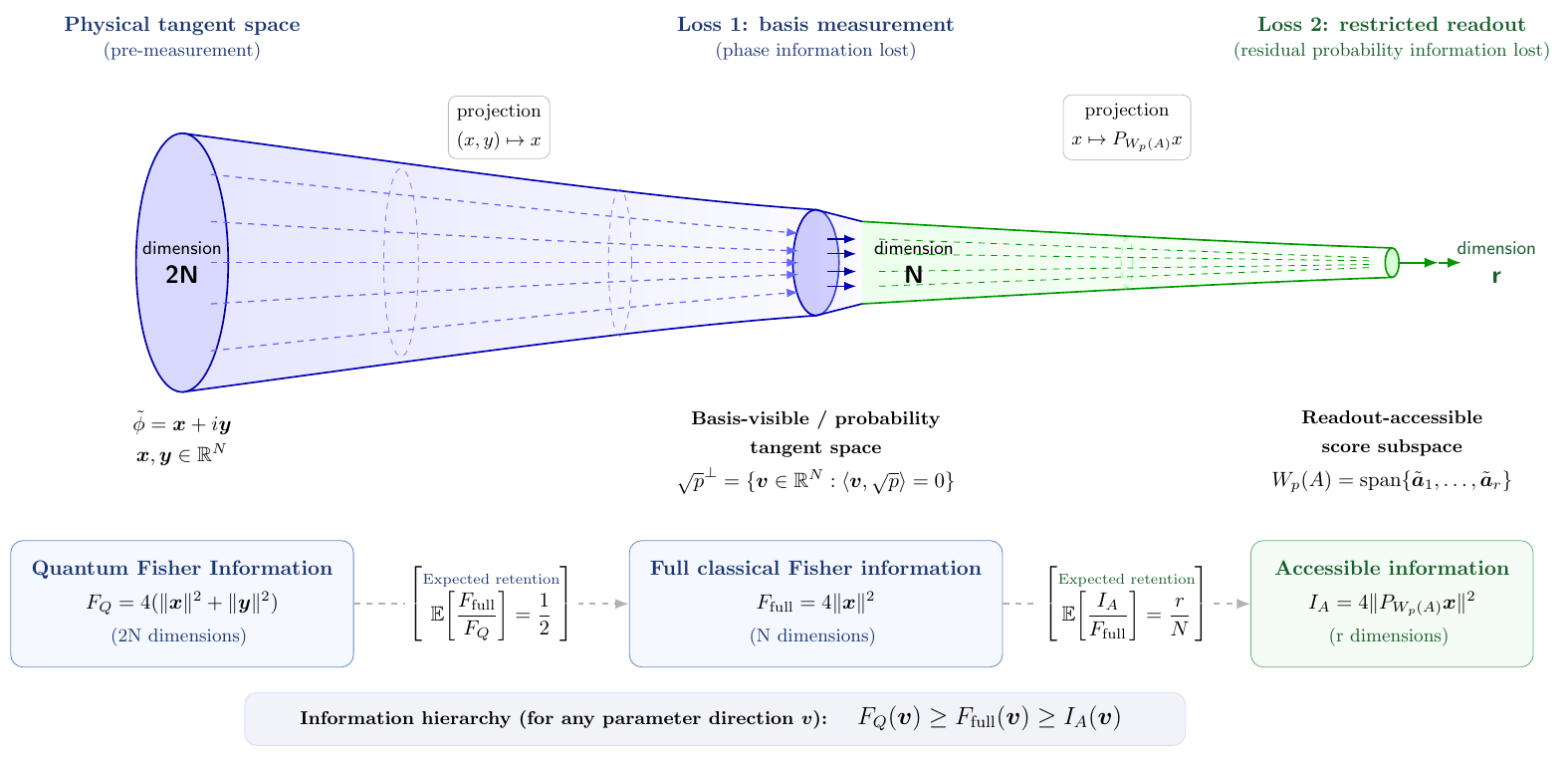}

 \caption{\textbf{Two-stage geometric information loss in QML.} Quantum information in the physical tangent space undergoes two successive orthogonal projections: (i) basis measurement discards phase directions; (ii) restricted readout discards probability-score directions orthogonal to the chosen observables.}
 \label{fig:geometry}
\end{figure*}

To make the effect of classical compression concrete, consider first a marginal record.  For a subset $S$ of the $n$ qubits, $p_S(z_S)$ denotes the probability of the restricted bit string $z_S$, and
\begin{equation}
 p_S(z_S)=\sum_{z_{\bar S}}p(z_S,z_{\bar S}),\qquad
 F_S(v)=\sum_{z_S}\frac{(\partial_vp_S(z_S))^2}{p_S(z_S)}.
 \label{eq:marginal_fi}
\end{equation}
The best marginal of order at most $k$ is
\begin{equation}
 F^{(k)}(v)=\max_{|S|\le k}F_S(v).
 \label{eq:best_marginal}
\end{equation}

\begin{proposition}[Marginal data processing]
If $S\subseteq T$, then $F_S(v)\le F_T(v)$.  Consequently,
\begin{equation}
 F^{(1)}\le F^{(2)}\le\cdots\le F^{(n)}=\Ffull\le\FQ.
 \label{eq:hierarchy}
\end{equation}
\end{proposition}

\begin{proof}
Let $s_T=\partial_v\log p_T$ be the score.  Marginalization gives
$s_S=\E[s_T|Z_S]$.  Conditional Jensen implies
$\E[s_S^2]\le\E[s_T^2]$, proving the classical inequalities.  The last inequality is the Braunstein--Caves measurement bound~\cite{braunstein1994}.
\end{proof}

The hierarchy orders the resolution of the record but does not predict the size of the gaps between successive marginals.  We will refer to the smallest $k$ that recovers a prescribed fraction of $F_{\rm full}$ as the \emph{marginal recovery order}; a large value of $k$ does not, by itself, identify an irreducible $k$-body correlation.

A marginal is only one possible compression of the bitstring record.  To treat arbitrary diagonal readouts at once, let $\A$ be a finite-dimensional real vector space of functions $c:\{0,1\}^n\to\mathbb R$ that contains the constant function.  Each $c$ is the eigenvalue function of a diagonal observable.  We define the directional information accessible to this whole readout space by
\begin{equation}
 \Iacc_{\A}(v)=
 \sup_{c\in\A:\,\Var_p(c)>0}
 \frac{\big(\partial_v\E_p[c]\big)^2}{\Var_p(c)}.
 \label{eq:I_A}
\end{equation}
The variance in the denominator removes an arbitrary rescaling of the observable.  In the many-shot local limit, the same covariance-normalized quadratic form appears as the information carried by empirical estimating functions~\cite{godambe1960}.

Associate to every readout the centered weighted vector
\begin{equation}
 w_c=\sqrt p\,(c-\E_p c)\in \sqrt p^{\perp},
 \label{eq:weighted_feature}
\end{equation}
and let
\begin{equation}
 W_p(\A)=\{w_c:c\in\A\},\qquad r=\dim W_p(\A).
 \label{eq:centered_rank}
\end{equation}
We call $r$ the \emph{centered readout rank}.  It counts the independent, nonconstant score directions that the readout can represent; with full support and no additional redundancy, $r=\dim\A-1$.

\begin{proposition}[Score projection]
For the computational-basis tangent in \cref{eq:q},
\begin{equation}
 \Iacc_{\A}(v)=4\|P_{W_p(\A)}x\|^2,
 \label{eq:projection}
\end{equation}
where $P_{W_p(\A)}$ is the Euclidean orthogonal projector.
\end{proposition}

\begin{proof}
Since $\sum_zq_z=0$,
\begin{equation*}
 \partial_v\E_p[c]=\sum_zc_zq_z
 =2\langle w_c,x\rangle,
 \qquad \Var_p(c)=\|w_c\|^2.
\end{equation*}
Taking the largest squared normalized inner product over $w_c\in W_p(\A)$ yields Eq.~\eqref{eq:projection}.
\end{proof}

The proposition turns readout accessibility into an ordinary projection problem.  The centered observables span a subspace of the full score space, and the best normalized response is the squared length of the probability tangent projected onto that subspace.  A readout can therefore be insensitive even when the state and the complete measurement distribution both move appreciably.

For $\A_S$, the space of all functions of a subset $S$, the score-projection identity is exactly
\begin{equation}
 \Iacc_{\A_S}(v)=F_S(v).
 \label{eq:marginal_projection}
\end{equation}
Thus marginal Fisher information is not a separate construction: it is the special case in which every diagonal statistic on $S$ is allowed.

For later use, every individual observable $c\in\A$ obeys
\begin{equation}
 |\partial_v\E_p[c]|^2\le \Var_p(c)\Iacc_{\A}(v),
 \label{eq:gradient_bound}
\end{equation}
with equality when the centered $w_c$ is proportional to $P_{W_p(\A)}x$.

The same tangent span controls a differentiable classical head.  Let
$m_a(\bm\theta)=\E_p[c_a]$ for a basis $\{c_a\}$ of $\A$, and let
$L(\bm\theta)=f(m_1,\ldots,m_r)$ be any differentiable scalar head.  The chain rule gives
\begin{align}
 \partial_v L
 &=\sum_a\frac{\partial f}{\partial m_a}\partial_v m_a
 =\partial_v\E_p[c_{\rm eff}],\nonumber\\
 c_{\rm eff}&=\sum_a\frac{\partial f}{\partial m_a}c_a\in\A.
 \label{eq:nonlinear_head}
\end{align}

A nonlinear classical head may choose a data-dependent linear combination of the measured features at each operating point, but it cannot create a new first-order quantum direction outside their span.  This observation will connect the readout rank to QML gradients in Section~\ref{sec:qml}.

\section{Haar tangent frames and finite-depth circuits}
\label{sec:haar}

Draw $|\psi\rangle$ from the Haar measure and, conditional on that state, draw a unit vector $|u\rangle$ uniformly from the complex subspace orthogonal to $|\psi\rangle$.  The pair $(\psi,u)$ is then a Haar-random orthonormal complex two-frame, equivalently the first two columns of a Haar unitary~\cite{mezzadri2007}.  Setting $|\phi\rangle=\sigma|u\rangle$ for an arbitrary tangent norm $\sigma>0$, the $2N$ real coordinates $(x,y)$ in \cref{eq:xy} are, conditional on $\psi$, uniformly distributed on the sphere of radius $\sigma$ in $\mathbb R^{2N}$.
We write $\operatorname{Beta}(a,b)$ for the Beta distribution with shape parameters $a$ and $b$.

The probability law used below is the standard finite-dimensional
hyperspherical projection law: if a vector is uniformly distributed
on a real sphere in dimension $n$, then the squared norm of its
projection onto a fixed $m$-dimensional subspace follows
$\mathrm{Beta}(m/2,(n-m)/2)$. This law has recently also been used
to derive exact subsystem-probability statistics of Haar-random
pure states from the projected central limit theorem
\cite{wang2026geometric}. Our use of it concerns a different object:
the joint state--tangent frame and its successive projections onto
the probability-tangent and readout-score subspaces.

\begin{theorem}[Full fixed-basis law]
For a Haar state--tangent frame,
\begin{equation}
 \frac{\Ffull}{\FQ}\sim
 \operatorname{Beta}\!\left(\frac N2,\frac N2\right),
 \qquad
 \E\frac{\Ffull}{\FQ}=\frac12.
 \label{eq:full_beta}
\end{equation}
\end{theorem}

\begin{proof}
Represent the squared coordinates of a uniform sphere vector by independent chi-square variables.  Then $\|x\|^2\sim\chi_N^2$ and $\|y\|^2\sim\chi_N^2$ up to a common normalization.  Their ratio in \cref{eq:amplitude_phase} is Beta distributed.
\end{proof}

An isotropic tangent has no preferred orientation between the amplitude and phase sectors.  The fixed basis therefore retains one half of the QFI on average.  This loss remains of order one as the number of qubits grows, so it cannot explain the exponential suppression produced by a low-rank readout.

\begin{theorem}[Readout-rank law]
Let $\A$ have centered rank $r$ on the support of $p$.  For a Haar state--tangent frame,
\begin{align}
 \frac{\Iacc_{\A}}{\Ffull}
 &\sim \operatorname{Beta}\!\left(\frac r2,\frac{N-r}{2}\right),
 \label{eq:rank_beta_full}\\
 \frac{\Iacc_{\A}}{\FQ}
 &\sim \operatorname{Beta}\!\left(\frac r2,\frac{2N-r}{2}\right).
 \label{eq:rank_beta_q}
\end{align}
Moreover, $\Iacc_{\A}/\Ffull$ and $\Ffull/\FQ$ are independent, and
\begin{equation}
 \E\frac{\Iacc_{\A}}{\Ffull}=\frac rN,
 \qquad
 \E\frac{\Iacc_{\A}}{\FQ}=\frac r{2N}.
 \label{eq:rank_means}
\end{equation}
\end{theorem}

\begin{proof}
Conditional on $\psi$, the subspace $W_p(\A)$ is fixed and has real dimension $r$.  Rotational invariance lets us decompose the tangent norm into three independent chi-square components:
\begin{equation*}
 A=\|P_Wx\|^2,\quad B=\|(I-P_W)x\|^2,\quad C=\|y\|^2,
\end{equation*}
with degrees of freedom $r$, $N-r$, and $N$.  Equations~\eqref{eq:projection} and \eqref{eq:amplitude_phase} give
\begin{align*}
 \frac{\Iacc_{\A}}{\Ffull}&=\frac A{A+B}, &
 \frac{\Ffull}{\FQ}&=\frac{A+B}{A+B+C},\\
 \frac{\Iacc_{\A}}{\FQ}&=\frac A{A+B+C}.&&
\end{align*}
Applying the finite-dimensional hyperspherical projection law to
the three orthogonal blocks of dimensions $r$, $N-r$, and $N$
yields the stated Beta distributions. The beta--gamma algebra also
shows that $A/(A+B)$ is independent of $A+B$, and therefore of
$(A+B)/(A+B+C)$.
\end{proof}

The theorem is a finite-size dimension-counting law, including its fluctuations.  Once the tangent orientation is isotropic, the identity of an observable matters only through the centered subspace spanned by the whole readout.  Redundant or overlapping features do not help beyond the rank that they add.

Two choices of readout make this statement useful for QML.  First, if $S$ is a fixed subset of $k$ qubits and the readout retains every diagonal function of its marginal outcome, then its centered rank is $2^k-1$.

\begin{corollary}[Fixed marginal]
For a fixed subset $S$ of $k$ qubits, $r=2^k-1$ and
\begin{equation}
 \frac{F_S}{\Ffull}\sim
 \operatorname{Beta}\!\left(\frac{2^k-1}{2},
 \frac{N-(2^k-1)}{2}\right).
 \label{eq:marginal_beta}
\end{equation}
In particular, $\E[F_S/\Ffull]=(2^k-1)/(2^n-1)$.
\end{corollary}

The best order-$k$ marginal is selected from a polynomial number of such subspaces.  For any threshold $t\in(0,1)$, a union bound gives
\begin{equation}
 \Prob\!\left[\frac{F^{(k)}}{\Ffull}\ge t\right]
 \le \sum_{j=1}^k\binom nj
 \Prob\!\left[B_{j,N}\ge t\right],
 \label{eq:union_bound}
\end{equation}
where $B_{j,N}\sim\operatorname{Beta}((2^j-1)/2,(N-2^j+1)/2)$.  For fixed $k$ and fixed $t>0$, the polynomial number of subsets cannot overcome concentration in the exponentially large score space.

Second, one may retain all low-weight diagonal Pauli expectations simultaneously.  For a subset $A\subseteq[n]$, define the Walsh function
\begin{equation}
 \chi_A(z)=(-1)^{\sum_{i\in A}z_i},
 \label{eq:walsh}
\end{equation}
which are the eigenvalue functions of computational-basis Pauli strings $Z_A=\prod_{i\in A}Z_i$.  Let
\begin{equation}
 \A_{\le k}=\operatorname{span}\{\chi_A:|A|\le k\}.
 \label{eq:Ak}
\end{equation}
This is the joint feature space containing \emph{all} diagonal Pauli strings through weight $k$, not the best single marginal.

\begin{corollary}[Aggregate low-order readout]
On the full Boolean cube,
\begin{equation}
 r_k=\sum_{j=1}^k\binom nj,
 \label{eq:rk}
\end{equation}
and the laws \eqref{eq:rank_beta_full}--\eqref{eq:rank_beta_q} hold with $r=r_k$.  Writing $\Iacc_{\le k}\equiv\Iacc_{\A_{\le k}}$, for fixed $k$,
\begin{equation}
 \E\frac{\Iacc_{\le k}}{\Ffull}
 =\frac{\sum_{j=1}^k\binom nj}{2^n-1}
 =\Theta\!\left(\frac{n^k}{2^n}\right).
 \label{eq:aggregate_scaling}
\end{equation}
\end{corollary}

Using every low-weight string is substantially stronger than choosing one marginal because the rank now grows polynomially with $n$.  The ambient score dimension, however, grows as $2^n-1$.  The ratio still closes exponentially for every fixed $k$.  The issue is therefore not the number of local features by itself, but the dimension they occupy inside the full statistical record.

The preceding laws concern the joint orientation of the state and its tangent.  They do not follow from randomness of the state alone.  A circuit produces the required joint frame exactly in the following limiting construction.

\begin{proposition}[Independent Haar suffix]
Suppose a parameterized circuit can be cut after the parameter action as
\begin{equation}
 |\psi_\theta\rangle=U_{\rm post}|a_\theta\rangle,
 \label{eq:suffix}
\end{equation}
where $U_{\rm post}$ is an independent Haar unitary.  If the horizontal tangent of $|a_\theta\rangle$ is nonzero, then the normalized output state and horizontal tangent form a Haar-random orthonormal two-frame.  Consequently, all laws in \cref{sec:haar} hold exactly.
\end{proposition}

\begin{proof}
Normalize the horizontal pair $(|a\rangle,|b\rangle)$.  Applying the same independent Haar unitary sends every fixed orthonormal pair to the Haar measure on the complex Stiefel manifold.  Unitary evolution preserves horizontality and tangent norm.
\end{proof}

Only the gates after the differentiated parameter rotate its tangent relative to the measurement basis.  An early parameter can therefore acquire a long scrambling suffix, whereas a late parameter cannot.  This predicts that agreement with the Haar law should depend on the available scrambling suffix: early parameters have longer post-parameter evolution, whereas late parameters retain shorter suffixes and may remain anisotropic or even fixed-basis blind.

Local random circuits approach unitary designs with depth~\cite{harrow2009,brandao2016,haferkamp2022}, motivating the Haar frame as a candidate deep-circuit limit.  However, a state design or even finite polynomial moment matching does not automatically prove \cref{eq:rank_beta_full}.  Fisher information contains $1/p_z$ and is sensitive to small probabilities, while the theorem requires the joint distribution of the state and its tangent.  We therefore study the finite-size approach of finite-depth VQCs to the Haar prediction numerically, without inferring asymptotic convergence from a frame potential or a finite-design condition. An explicit finite-depth error bound for
\cref{eq:rank_beta_full,eq:rank_beta_q} would require more than
closeness of the output state to a finite $t$-design: one must control
the joint state--tangent frame together with the lower tail of the
measurement probabilities entering the Fisher ratios. We leave such a
joint-frame approximation theorem open.

Symmetries require one further qualification.  If the measured distribution is supported on only $M<D$ bit strings, all scores and readout functions must first be restricted to that support.  The classical tangent dimension then becomes
\begin{equation}
 N_{\rm supp}=M-1.
 \label{eq:support_dim}
\end{equation}
The relevant readout rank is also recomputed after restricting the functions to the support.  This correction is necessary but not sufficient: the tangent orientation may remain anisotropic within the symmetry sector.

For a half-filled $U(1)$ circuit, $M=\binom n{n/2}$.  On this slice the low-degree Walsh functions are linearly dependent; the harmonic analysis of the Boolean slice~\cite{filmus2016} gives, for the orders tested here, the centered rank
\begin{equation}
 r_k^{U(1)}=\binom nk-1,\qquad k\le n/2.
 \label{eq:u1_rank}
\end{equation}
The numerical projection uses the actual Gram rank, so no analytic rank formula is assumed by the test.  More importantly, correcting the support and rank does not force the tangent to be isotropic inside the symmetry sector.  The number-conserving circuits below test precisely this remaining condition.

\section{Numerical evidence}
\label{sec:results}

All results are exact-statevector calculations with analytic tangent propagation.  Rotation angles are independent and uniform on $[-\pi,\pi]$.  Unless a coordinate direction is specified, $v$ is a normalized Gaussian vector in the full parameter space. Each reported tangent is generated from an independently seeded circuit realization: the circuit parameters, tangent direction, and any randomized architecture are resampled for every job.  We propagate $|\psi\rangle$ and $|\dot\psi_v\rangle$ together and project the latter horizontally at the output.  This avoids finite-difference noise.

Five nonconserving circuit families provide independent realizations of the same test: $R_yR_z$ layers with a nearest-neighbor CZ line; arbitrary single-qubit $SU(2)$ layers with either a CNOT line, a CZ ring, or a fresh random CZ matching; and $SU(2)$ layers with Haar-random two-qubit brickwork gates.  We call these the \emph{generic} families.  A sixth circuit, built from $R_z$ rotations and parameterized XY brickwork gates at half filling, conserves particle number and serves as the $U(1)$ control.

Writing $d$ for the number of circuit layers, the depth--size experiment uses $n\in\{6,8,10,12,14\}$ and
$d/n\in\{0.5,1,2,4,6\}$.  The primary block has 30 random directions per cell.  At $d=6n$ and $n=12,14$, 170 additional directions per cell raise the sample count to 200.  Coordinate directions are sampled from early, middle, and late layer thirds at $d=2n$ and $6n$.  Together these blocks contain 9420 circuit tangents.  A separate exact hierarchy block contains 120 circuits at $n=8,12$, and the Haar control contains 1500 independent frames.

For each tangent we evaluate $F_Q$, $F_{\rm full}$, a fixed marginal, and the best marginal over every subset for $k=1,2,3,n-1$.  Exact enumeration over all $k$ is performed in the hierarchy block.  We define the fitted decay exponent $\gamma_k$ over the simulated range by
\begin{equation}
 \operatorname{median}\!\left(F^{(k)}/F_{\rm full}\right)
 \propto 2^{-\gamma_k n},
 \label{eq:fit}
\end{equation}
using $n=6,8,10,12,14$ and 1000 stratified bootstrap resamples of the independently generated circuit--tangent jobs within each family and depth \cite{efron1979}.

To test whether the observed decay is specifically exponential
rather than merely decreasing, let
\[
m_k(n)
\equiv
\operatorname{median}
\left[
\frac{F^{(k)}}{F_{\mathrm{full}}}
\right].
\]
We compare the three models
\[
\log_2 m_k(n)=a-\gamma n,
\]
\[
\log_2 m_k(n)=a-\beta\log_2 n,
\]
and
\[
\log_2 m_k(n)=a-\gamma n+\beta\log_2 n,
\]
corresponding respectively to exponential, polynomial, and
exponential-times-polynomial decay. Model preference is
determined by the smallest leave-one-size-out root-mean-square
prediction error in $\log_2 m_k(n)$.
An independent aggregate experiment evaluates the joint space $\A_{\le k}$ of all strings through weight $k$ on 1500 new Haar frames and 600 new circuit tangents at depth $d=6n$.  The projection is computed from the full covariance matrix of the Walsh features, including all correlations between them; the equivalent pseudoinverse formula is given in Appendix~\ref{app:support}. Agreement with the exact Haar distributions is quantified using the KS distance, the Cram\'{e}r--von Mises statistic, and the Wasserstein distance of the probability-integral-transformed samples from the uniform distribution. Holm-corrected~\cite{holm1979} p-values are reported as secondary hypothesis-testing diagnostics rather than as measures of distributional proximity.
For finite-depth circuits, we compare both the predicted means and the full distributions, since agreement with the rank scaling does not imply an exact Beta law.

\subsection{Haar benchmark}

Figure~\ref{fig:haar} verifies the two projections separately.  For $F_{\rm full}/F_Q$, all five sizes are consistent with the exact Beta law after Holm correction (smallest adjusted $p=0.893$), and the empirical means range from $0.4995$ to $0.5008$.  For fixed marginals, all 20 size--order tests are consistent with the corresponding exact Beta laws (smallest adjusted KS $p=0.271$).  The independent aggregate experiment is likewise consistent with all 15 exact Beta predictions (smallest adjusted KS $p=0.691$).  At $n=14$, its empirical/theoretical mean ratios are $1.0148$, $1.0006$, and $1.0002$ for $k=1,2,3$.

\begin{figure*}[t]
 \centering
 \includegraphics[width=0.99\textwidth]{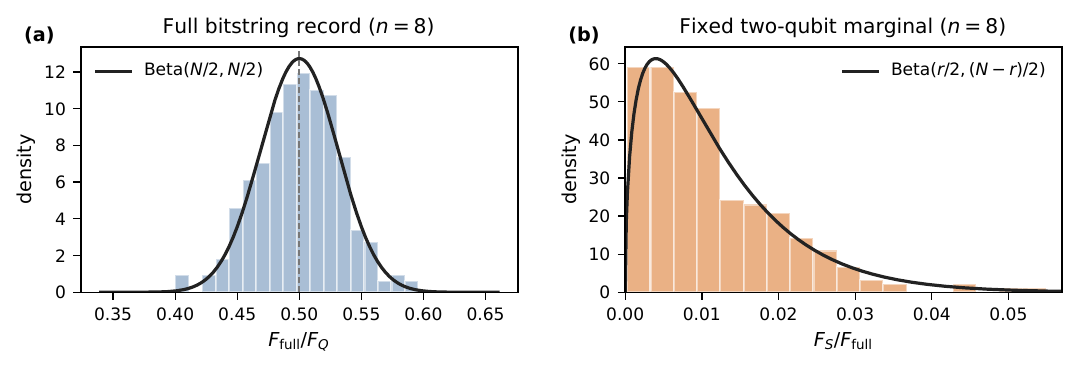}
 \caption{The two exact Haar projections tested at $n=8$.  (a) The fraction of quantum Fisher information visible in the full computational-basis record follows \cref{eq:full_beta}.  (b) The fraction of that record visible in a fixed two-qubit marginal follows \cref{eq:marginal_beta}.  Each histogram contains 300 independent frames; the black curves have no fitted parameters.}
 \label{fig:haar}
\end{figure*}

These tests do more than check a scaling exponent.  They validate the finite-size distribution, the weighted score projection, and the rank calculation before any circuit claim is made.

\subsection{Circuit crossover}

At $n=14$ and $d=6n$, the five generic families have mean $F_Q$ between $0.992$ and $0.996$ and mean $F_{\rm full}$ between $0.486$ and $0.494$.  Their ratios $F_{\rm full}/F_Q$ lie between $0.4898$ and $0.4967$, close to the Haar mean $1/2$ [\cref{fig:crossover}(a)].  The circuits have therefore not become insensitive over the
tested range: both the total tangent norm and the full-record
classical information remain numerically of order unity.

The best marginal tells a different story.  At the same point, the median fractions retained by the best one-, two-, and three-qubit subsets are respectively
\begin{equation}
\begin{aligned}
 k=1:&\quad(2.12\!\text{--}\!2.23)\times10^{-4},\\
 k=2:&\quad(5.79\!\text{--}\!6.16)\times10^{-4},\\
 k=3:&\quad(1.11\!\text{--}\!1.17)\times10^{-3}.
 \label{eq:best_n14}
\end{aligned}
\end{equation}
Across the five generic families, the fitted best-subset exponents are $0.921 - 0.967$ for $k=1$, $0.909 - 0.951$ for $k=2$, and $0.920 - 0.932$ for $k=3$. Every bootstrap interval excludes zero. To distinguish exponential decay from a power law over the simulated range, we compared exponential, polynomial, and exponential-times-polynomial models by leave-one-size-out prediction in $log_2$ space. Among the 15 generic family--order combinations at $d=6n$, a pure exponential model is preferred in 10 cases and an exponential with a polynomial prefactor in the remaining 5; a pure polynomial model is never preferred. We therefore describe the observed behavior as consistent with exponential decay over $n\in\{6,8,10,12,14\}$, rather than as an asymptotic scaling result. Across the full depth grid, the fitted exponent increases toward the Haar value as $d/n$ grows.
This is evidence of a depth-size crossover, not an assumption of asymptotic Haar behavior.

Distributional effect sizes show the same crossover.
To avoid confounding distributional distance with unequal
sample counts, this comparison uses only the primary block
of 30 independently generated samples per
family--size--depth cell. Averaged over the five generic
families and system sizes, the mean KS distances for
$k=1,2,3$ decrease from $0.535$, $0.715$, and $0.761$
at $d/n=0.5$ to $0.163$, $0.167$, and $0.159$ at
$d/n=6$. The corresponding Beta-PIT Wasserstein distances
decrease from $0.287$, $0.372$, and $0.391$ to $0.058$,
$0.066$, and $0.062$. Thus the evidence is not based only
on failure to reject a null distribution: at fixed sample
count, the measured distance to the Haar prediction
decreases systematically with depth.

\begin{figure*}[t]
 \centering
 \includegraphics[width=0.99\textwidth]{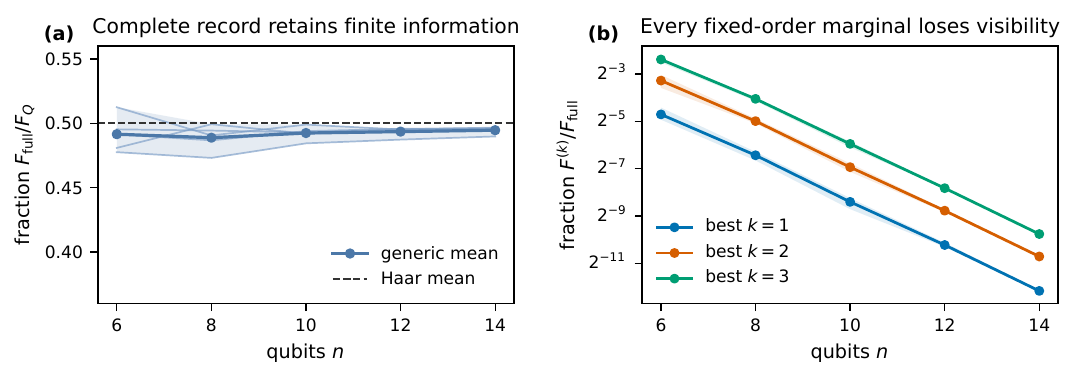}
 \caption{State sensitivity and readout visibility separate in deep generic circuits.  (a) At depth $d=6n$, the complete bitstring record retains approximately one half of the quantum Fisher information as $n$ grows.  (b) Over the same circuits, the best marginal on one, two, or three qubits retains a fraction consistent with exponential decay over the simulated range $n\in\{6,8,10,12,14\}$.  Bands show the range across the five nonconserving circuit families.}
 \label{fig:crossover}
\end{figure*}

Exact hierarchy enumeration reinforces the fixed-order result.  Let $k_{0.9}$ denote the smallest marginal order that recovers at least $90\%$ of $F_{\rm full}$.  Every generic circuit in the $n=8,12$ hierarchy block has $k_{0.9}=n$.  At $n=14,d=6n$, all 1000 generic confirmatory samples also satisfy $F^{(n-1)}<0.9F_{\rm full}$ and therefore have $k_{0.9}=n$.  This statement concerns marginal recovery; it should not be read as proof of irreducible $n$-body correlation.

\subsection{Readout rank and symmetry}

The strongest practical test is the aggregate readout $\A_{\le k}$.  Figure~\ref{fig:aggregate}(a) compares the circuit means with the parameter-free rank prediction.  At $n=14$, the centered ranks are $14$, $105$, and $469$, giving theoretical full-record fractions $0.0008545$, $0.0064091$, and $0.0286272$.  Across the five generic families, the observed ranges are shown in \cref{tab:aggregate}.

For the generic full-support families, the Walsh functions through
weight $k$ are linearly independent on the full Boolean cube. Hence
the centered rank is the combinatorial value $r_k$ in
\cref{eq:rk} and does not itself grow with circuit depth.
The observed depth crossover therefore reflects the orientation of the
probability tangent relative to this fixed feature span. The $U(1)$
family is qualitatively different because restriction to the
fixed-charge support creates exact linear dependencies.

\begin{table*}[t]
\centering
\caption{Information retained by the joint space of all diagonal Pauli strings through weight $k$ at $n=14$ and $d=6n$.  The intervals span the five generic circuit families.}
\label{tab:aggregate}
\begin{tabular}{ccccc}
\toprule
$k$ & rank $r_k$ & predicted / $F_{\rm full}$ & observed / $F_{\rm full}$ & observed / $F_Q$\\
\midrule
1 & 14  & $0.085\%$ & $0.076$--$0.096\%$ & $0.037$--$0.047\%$\\
2 & 105 & $0.641\%$ & $0.584$--$0.649\%$ & $0.286$--$0.322\%$\\
3 & 469 & $2.863\%$ & $2.725$--$2.902\%$ & $1.334$--$1.440\%$\\
\bottomrule
\end{tabular}
\end{table*}

Thus all 14 one-body $Z$ features jointly retain only $0.076$--$0.096\%$ of the full-record information.  Adding every one- and two-body string raises this to $0.584$--$0.649\%$; all 469 strings through weight three reach $2.725$--$2.902\%$.  The mean/rank-law ratio stays within $0.894$--$1.119$ for $k=1$, $0.912$--$1.013$ for $k=2$, and $0.952$--$1.014$ for $k=3$ [\cref{fig:aggregate}(b)].

Across the $75$ generic family--size--order comparisons, the mean KS distance to the parameter-free aggregate Beta prediction is $0.203$, with order-resolved means $0.188$, $0.204$, and $0.215$ for $k=1,2,3$. After Holm correction, none of the comparisons rejects the Beta distribution at level $0.05$, with a smallest adjusted p value of $0.167$. Nevertheless, systematic mean offsets, especially in the $R_yR_z$--CZ family, preclude a claim of an exact finite-depth Beta law. The robust circuit-level conclusion is increasing finite-size agreement with the rank-controlled mean hierarchy; evidence for the full distribution remains partial.

\begin{figure*}[t]
 \centering
 \includegraphics[width=0.99\textwidth]{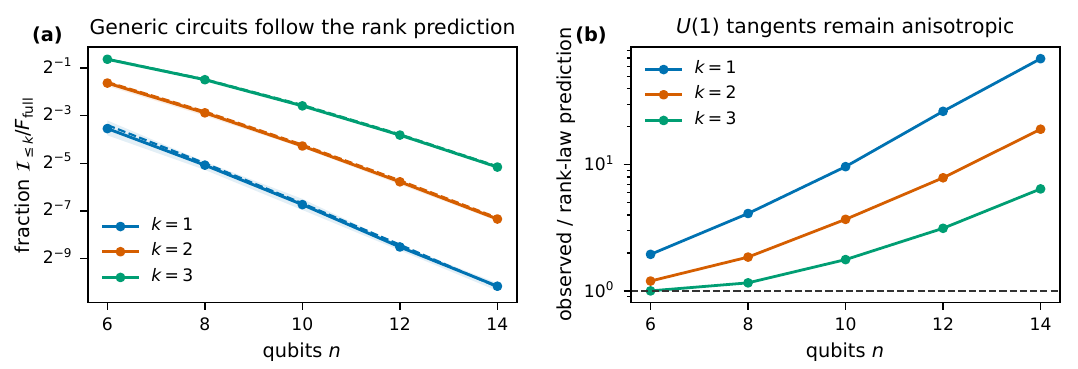}
 \caption{Scope of the readout-rank law.  (a) For the five generic circuit families, the mean information retained jointly by all $Z$ strings through weight $k$ follows the rank prediction (dashed curves).  (b) The number-conserving $U(1)$ circuit retains far more low-weight information than predicted, even after the support size and the actual Gram rank are used.}
 \label{fig:aggregate}
\end{figure*}

To separate the dimension factor from residual orientation
effects, define
\begin{equation}
\rho_k
\equiv
\frac{\Iacc_{\leq k}/F_{\mathrm{full}}}
     {r_k/N_{\mathrm{supp}}}.
\label{eq:normalized-rank-ratio}
\end{equation}
The isotropic rank law predicts
$\mathbb{E}[\rho_k]=1$. At $n=14$, the generic-family mean for $k=1$ is approximately unity, whereas the $U(1)$ value is $69.2$; the corresponding $U(1)$ enhancements are 19.1 and $6.43$ for $k=2,3$.

The generic agreement is not automatic.  At $n=14$ and $d=6n$, the $U(1)$ family still has a nonzero state tangent ($F_Q=0.778$) and a nonzero full-record response ($F_{\rm full}=0.367$), but it does not approach the support-adapted rank law.  Its fixed-order best-subset exponents are only $0.206$, $0.249$, and $0.253$ for $k=1,2,3$.

The aggregate discrepancy is larger.  At $n=14$, the half-filled support contains 3432 bit strings, and the measured centered Gram ranks are $13$, $90$, and $363$.  These values predict retained fractions $0.003789$, $0.026231$, and $0.10580$.  The observed fractions are instead $0.26217$, $0.50098$, and $0.68008$, corresponding to enhancements by factors $69.2$, $19.1$, and $6.43$ [\cref{fig:aggregate}(b)].  Because both the support and rank have already been corrected, the discrepancy shows that rank alone is insufficient and that tangent orientation inside the charge sector remains essential.

There is also a structural recovery effect: in a fixed-Hamming-weight record, any $n-1$ bits determine the last bit, so $F^{(n-1)}=F_{\rm full}$ exactly.  The $U(1)$ result must therefore not be described as the same extensive marginal-recovery phenomenon with a different exponent.

Conserved random circuits possess diffusive slow modes and require parametrically greater depth to approach design behavior than unconstrained circuits~\cite{rakovszky2018,hearth2025}.  Our depth $d=O(n)$ is consequently too short to justify a sector-Haar assumption.  This makes the failure of a sector-Haar description physically plausible, but it does not quantitatively derive the enhancement factors in \cref{fig:aggregate}(b); those factors remain empirical measures of tangent anisotropy in the present work.

\section{Implications and limits}
\label{sec:qml}

Many variational classifiers use Pauli expectation values as
quantum features~\cite{havlicek2019,schuld2020,perezsalinas2020}.
Consider a model that retains every computational-basis Pauli
string $Z_A$ with weight $|A|\leq k$, where $A$ labels the
qubits on which the string acts. A differentiable classical
head has the form
\begin{equation}
 m_A(\bm\theta)
 =
 \langle Z_A\rangle_{\bm\theta},
 \qquad
 |A|\leq k,
 \qquad
 L=f(\{m_A\}).
 \label{eq:qml_head}
\end{equation}

At a fixed operating point, the chain rule combines the
derivatives of the measured moments into the response of a
single effective diagonal statistic. Using the Walsh functions
defined in \cref{eq:walsh}, its eigenvalue function is
\[
 c_{\rm eff}(z)
 =
 \sum_{|A|\leq k}
 \frac{\partial f}{\partial m_A}\,
 \chi_A(z)
 \in
 \A_{\leq k}.
\]
This is the specialization of \cref{eq:nonlinear_head} to the
aggregate low-order readout space in \cref{eq:Ak}. The loss
gradient along the parameter direction $v$ therefore satisfies
\begin{equation}
 |\partial_vL|^2
 \leq
 \Var_p(c_{\rm eff})\,
 \Iacc_{\leq k}(v).
 \label{eq:qml_bound}
\end{equation}

The nonlinear head can change the coefficients of
$c_{\rm eff}$ at each operating point, but it cannot enlarge
the first-order span supplied by the measured quantum
features. Accordingly, $\Iacc_{\A}$ is an optimal local ceiling
over the readout span, not the information necessarily attained
by a particular trained model. A concrete classical head may
access less information, but it cannot exceed this
covariance-normalized ceiling without enlarging the measured
feature span.

For an isotropic tangent, combining
\cref{eq:qml_bound} with the readout-rank law gives
\begin{equation}
 \E\left[
 \frac{|\partial_vL|^2}
 {\Var_p(c_{\rm eff})F_Q}
 \right]
 \leq
 \E\left[
 \frac{\Iacc_{\leq k}}{F_Q}
 \right]
 =
 \frac{r_k}{2(2^n-1)}
 =
 \Theta\!\left(n^k2^{-n}\right).
 \label{eq:qml_scaling}
\end{equation}
Here the ratio is understood on the set where
$\Var_p(c_{\rm eff})>0$; when the variance vanishes,
\cref{eq:qml_bound} implies $\partial_vL=0$, and the bound is
trivial.

The $n=14$ values in \cref{tab:aggregate} illustrate the
scale. Even the maximum variance-normalized response available
within the joint span of all one-body features is only about
$0.04\%$ of the QFI along a generic deep-circuit direction.
Enlarging the span to include every diagonal Pauli string
through weight three raises this ceiling to only
$1.3$--$1.4\%$.

At the level of the isotropic mean and the observed
generic-circuit scaling, the readout-induced local-blindness
mechanism can be summarized as
\begin{equation}
\begin{aligned}
 &\underbrace{F_Q=O(1)}_{\text{state remains sensitive}}
 \quad\land\quad
 \underbrace{\frac{F_{\rm full}}{F_Q}=\Theta(1)}
 _{\text{basis record retains a finite fraction}}
 \\[1mm]
 &\hspace{13mm}\land\quad
 \underbrace{
 \E\left[
 \frac{\Iacc_{\leq k}}{F_Q}
 \right]
 =
 \Theta\!\left(n^k2^{-n}\right)}
 _{\substack{\text{fixed-order readout retains}\\\text{an exponentially small fraction}}}.
 \label{eq:three_conditions}
\end{aligned}
\end{equation}

Equation~\eqref{eq:qml_scaling} directly establishes
exponential suppression of the variance-normalized gradient.
It does not, by itself, establish the same scaling for the raw
gradient, because the classical head may rescale the effective
observable through its derivatives.

To make the additional condition explicit, let $\Sigma_m$ be
the covariance matrix of the retained features,
\[
 (\Sigma_m)_{AB}
 =
 \operatorname{Cov}_p(\chi_A,\chi_B),
 \qquad
 |A|,|B|\leq k.
\]
Writing $\nabla_m f$ for the vector of derivatives
$\partial f/\partial m_A$, one has
\[
 \Var_p(c_{\rm eff})
 =
 \nabla_m f^{\mathsf T}
 \Sigma_m
 \nabla_m f
 \leq
 \|\Sigma_m\|_{\rm op}
 \|\nabla_m f\|_2^2.
\]
Therefore, if the bounds
\[
 \|\nabla_m f\|_2
 \leq
 \operatorname{poly}(n),
 \qquad
 \|\Sigma_m\|_{\rm op}
 \leq
 \operatorname{poly}(n)
\]
hold uniformly over the ensemble, then
\[
 \Var_p(c_{\rm eff})
 \leq
 \operatorname{poly}(n).
\]
Combining this condition with \cref{eq:qml_scaling} yields
\[
 \E\left[
 \frac{|\partial_vL|^2}{F_Q}
 \right]
 \leq
 \operatorname{poly}(n)\,
 \frac{r_k}{2(2^n-1)}
 =
 \operatorname{poly}(n)\,2^{-n}
\]
for every fixed $k$. If $F_Q$ is additionally uniformly of
order unity, the raw squared gradient inherits the same
exponential factor.
We do not assume that every practical classical head satisfies these
polynomial bounds; they are sufficient conditions under which the
variance-normalized readout ceiling transfers to a raw-gradient scaling
statement.

This is a conditional readout-induced mechanism, not a
complete barren-plateau theorem. A barren plateau concerns the
distribution of a specified raw cost gradient over a parameter
ensemble. By contrast, \cref{eq:qml_bound,eq:qml_scaling}
provide a directional, variance-normalized ceiling; the
raw-gradient conclusion additionally requires regularity of
the classical head and feature covariance.

This viewpoint clarifies why the phrase ``measure every qubit
locally'' can be misleading. Measuring a complete bitstring is
global at the level of the classical record. Keeping only the
$n$ numbers $\langle Z_i\rangle$ is a rank-$n$ compression of
a score space of dimension $2^n-1$. A classifier that returns
one Pauli expectation per wire therefore implements a
low-rank readout even though every wire is measured
\cite{schuld2020,perezsalinas2020}.

The projection theorem is not exclusively quantum. Any
isotropic tangent in a finite probability simplex retains a
mean fraction $r/N$ when projected onto an $r$-dimensional
readout. The quantum contribution is the preceding
amplitude--phase split, which contributes an additional factor
$1/2$ relative to the QFI, together with the common unitary
evolution that rotates a state and its tangent as a joint
frame. The result is therefore an information-geometric law
with a quantum-circuit realization, not a uniquely quantum
impossibility statement.

The two information ratios also separate different departures from
isotropy. In a block-isotropic alternative where the direction
$x/\|x\|$ remains rotationally invariant inside the probability sector
but the radial balance between $\|x\|$ and $\|y\|$ is biased, the
conditional rank projection $\Iacc_{\A}/\Ffull$ retains the
$\operatorname{Beta}(r/2,(N-r)/2)$ law, whereas
$\Ffull/\FQ$ need not follow the symmetric
$\operatorname{Beta}(N/2,N/2)$ law. By contrast, anisotropy of the
direction of $x$ within the probability sector changes the readout
projection itself. General departures from isotropy therefore do not
correspond to a universal shift of the Beta parameters: amplitude--phase
imbalance and within-score orientation are distinct effects.

Several limits are essential. The analysis fixes the
computational basis and diagonal functions of its outcomes.
Basis rotations, randomized measurements, ancillas, and
general POVMs can change the accessible subspace
\cite{hotta2004,lu2012,huang2020,lu_sha_2025_random_cfim,
kolotouros2025_random_measurements}. In particular, randomized-basis
CFIM results for pure states show that Haar averaging recovers one half
of the QFIM and that the fluctuations concentrate with increasing
Hilbert-space dimension \cite{lu_sha_2025_random_cfim}. This is
complementary to the present result: here the basis is fixed, and the
second, rank-$r$ readout projection is the source of the fixed-order
exponential bottleneck. A resource comparison between multiple basis
settings and a larger fixed-basis readout must account for state
preparations, shots, basis settings, and retained statistics, and is
left open.

The exact distributions require joint state--tangent isotropy. A
Haar-distributed state or a finite $t$-design condition on the state
alone is insufficient, and we do not derive a convergence rate from
the tested finite-depth circuits to the Haar law. The proof is also
specific to pure-state projective geometry. Mixed-state SLD geometry
does not share the same equal-dimensional amplitude--phase
decomposition, and a purification-based extension would require an
additional treatment of gauge directions. We therefore make no
mixed-state rank-law claim here.

Nor do we claim to reconstruct the full tangent covariance or
to prove joint-frame isotropy from the circuit data. Instead,
we test several consequence-level signatures:
amplitude--phase balance, distributional distance to the Beta
laws, aggregate rank ratios, and parameter-location
dependence. Their joint agreement provides evidence for an
isotropic crossover, not a direct proof of isotropy.

Symmetry sectors additionally require support-restricted
scores and ranks, while nonregular zeros with nonzero
derivatives require a separate limiting analysis. The
numerical evidence is limited to noiseless statevector
simulations through 14 qubits. Finite-shot estimation turns
hybrid optimization into a stochastic procedure
\cite{sweke2020}, while local hardware noise can itself
suppress gradients with circuit depth~\cite{wang2021};
neither effect is included here.

The rank viewpoint also suggests adaptive readout expansion: when the
covariance-normalized response inside the current span is inadequate,
one may add observables that contribute new independent score
directions. The theorem quantifies the benefit of added rank under
isotropy, but it does not provide a polynomial-cost rule for discovering
useful new directions; that measurement-design problem is separate from
the dimension-counting law.

Finally, Fisher information describes an infinitesimal
response. It neither determines classification accuracy nor
specifies an optimization trajectory or finite-displacement
distinguishability. The $U(1)$ counterexample reinforces these
limits by showing that low-order tangent modes can survive
even in a large symmetry sector whose support dimension grows
exponentially.
\section{Conclusion}
\label{sec:conclusion}

Under joint state--tangent isotropy, a quantum model can carry an $O(1)$ state tangent and an $O(1)$ fixed-basis Fisher information while exposing only an exponentially small fraction to any fixed-order diagonal readout space.  The mechanism is not the disappearance of information.  It is the geometry of projecting an isotropic tangent from an exponentially large score space onto a polynomial-rank readout space.

For Haar state--tangent frames this statement is exact at finite size: the fixed basis selects an $N$-dimensional amplitude sector from a $2N$-dimensional quantum tangent, and a centered rank-$r$ readout selects $r$ directions from that sector.  The two fractions obey independent Beta laws.  Over the tested depth--size grid, five nonconserving circuit families show increasing finite-size agreement with this hierarchy, whereas a number-conserving family remains far from the isotropic prediction even after its support and readout rank are corrected.

The QML consequence is precise but limited.  A 14-qubit classifier feeding only the 14 values $\langle Z_i\rangle$ to its classical head has rank 14 inside a score space of dimension $2^{14}-1$ and, under isotropy, accesses on average only $8.55\times10^{-4}$ of $F_{\rm full}$.  Classical depth alone cannot enlarge that first-order
quantum-feature span; one must change the circuit regime,
measurement basis, readout rank, or symmetry-adapted
features.

Readout rank is therefore predictive only after the state--tangent frame has been sufficiently isotropized; outside that regime, accessible information also depends on the orientation of the probability tangent relative to the readout subspace.

\section*{Data Availability}
The complete sample-level CSV files, bootstrap outputs, analysis
notebooks, projection code, model-comparison code, and figure-generation
scripts supporting this study are archived on Zenodo
\cite{ait_haddou_2026_code_data}.

\clearpage
\onecolumngrid
\appendix

\section{Conditional Haar-frame derivation}
\label{app:haar}

We give the conditional construction used in \cref{sec:haar}.  Fix
$\psi_z=\sqrt{p_z}e^{i\vartheta_z}$ with $p_z>0$ and let $u$ be uniform on the complex unit sphere in $\psi^\perp$.  After phase alignment, $\tilde u_z=e^{-i\vartheta_z}u_z$ satisfies
\begin{equation}
 \sum_z\sqrt{p_z}\,\tilde u_z=0.
 \label{eq:complex_constraint}
\end{equation}
Choose a real orthogonal matrix whose first basis vector is $\sqrt p$.  In the remaining $N$ coordinates, the real and imaginary parts of $\tilde u$ form a uniform vector on $S^{2N-1}$.  The distribution is therefore independent of the particular probabilities $p_z$.

Let an orthogonal change of basis within $\sqrt p^\perp$ align the first $r$ coordinates with $W_p(\A)$.  If $g_i$ are independent standard normal variables, define
\begin{equation}
 A=\sum_{i=1}^{r}g_i^2,\quad
 B=\sum_{i=r+1}^{N}g_i^2,\quad
 C=\sum_{i=N+1}^{2N}g_i^2.
 \label{eq:chi_decomp}
\end{equation}
Then $A,B,C$ are independent chi-square variables with $r,N-r,N$ degrees of freedom.  Normalizing by $A+B+C$ produces the sphere vector but cancels from every information ratio.  This proves the three identities
\begin{equation}
\begin{aligned}
 \frac{\Iacc_{\A}}{\Ffull}&=\frac A{A+B}, &
 \frac{\Ffull}{\FQ}&=\frac{A+B}{A+B+C},\\
 \frac{\Iacc_{\A}}{\FQ}&=\frac A{A+B+C}.&&
 \label{eq:chi_ratios}
\end{aligned}
\end{equation}
Equivalently, $A$, $B$, and $C$ are independent Gamma variables
with common scale $2$ and shapes $r/2$, $(N-r)/2$, and $N/2$.
Therefore
\[
 U=\frac{A}{A+B}
 \sim \operatorname{Beta}\!\left(\frac r2,\frac{N-r}{2}\right)
\]
and the beta--gamma factorization implies that $U$ is independent of
$S=A+B$. Since $C$ is independent of $(A,B)$, it is also independent of
$(U,S)$; hence
\[
 V=\frac{S}{S+C}
 =\frac{\Ffull}{\FQ}
\]
is independent of $U=\Iacc_{\A}/\Ffull$. This proves the stated
independence. Finally, $A/(A+B+C)$ is the squared projection of the
same isotropic vector onto an $r$-dimensional subspace of the full
$2N$-dimensional real tangent space, giving
$\operatorname{Beta}(r/2,(2N-r)/2)$ for $\Iacc_{\A}/\FQ$.

\section{Zeros, support restriction, and Gram rank}
\label{app:support}

For a regular classical path, outcomes outside the support satisfy $p_z=q_z=0$ and can be removed.  Replacing $D$ by $M=|\supp p|$ gives a simplex tangent dimension $M-1$.  The weighted feature map \cref{eq:weighted_feature} may acquire additional linear dependencies on the restricted support, which is why $r$ must be computed after restriction.

The Gram formulation makes this explicit.  For a feature vector $c=(c_1,\ldots,c_m)$, let
\begin{equation}
 \Sigma_{ab}=\operatorname{Cov}_p(c_a,c_b),\qquad
 g_a=\partial_v\E_p[c_a].
 \label{eq:covariance}
\end{equation}
Then
\begin{equation}
 \Iacc_{\A}=g^\mathsf{T}\Sigma^+g,
 \label{eq:cov_pinv}
\end{equation}
where $\Sigma^+$ is the Moore--Penrose pseudoinverse, provided $g$ lies in the range of $\Sigma$, as it does for a regular path.  Adding the constant feature converts $\Sigma$ to the uncentered Gram matrix used in \cref{eq:gram_projection}, and its centered rank is one less than the Gram rank.

For the aggregate Walsh readout used numerically, the same formula takes a particularly simple form.  If $A$ and $B$ label Pauli strings, $A\triangle B$ denotes their symmetric difference, and
\begin{equation}
 G_{AB}=\E_p[\chi_{A\triangle B}],\qquad
 b_A=\partial_v\E_p[\chi_A],
 \label{eq:walsh_gram}
\end{equation}
then
\begin{equation}
 \Iacc_{\le k}=b^\mathsf{T}G^+b,
 \qquad r=\rank(G)-1.
 \label{eq:gram_projection}
\end{equation}
Here $G^+$ is the Moore--Penrose pseudoinverse.  The subtraction of one removes the constant feature.

If $p_z=0$ but $q_z\ne0$, the ordinary quadratic Fisher expression is nonregular and may diverge.  Such directions are excluded from the empirical scaling fits rather than assigned an arbitrary floor. The accounting below exhausts the 9420
crossover jobs, so no additional $p_z=0$, $q_z\ne0$ case occurs in
that block. We do not infer a universal prevalence of nonregular
directions from these finite ensembles.

\paragraph{Numerical checks and reproducibility.}
Of the 9420 crossover jobs, 9394 follow a regular probability path, 20 are exactly blind in the computational basis, and 6 have a vanishing physical tangent.  The largest state-norm error is below $3.4\times10^{-14}$, and no projected information exceeds $F_{\rm full}$ by more than $1.6\times10^{-15}$.  Late coordinate directions are noisier than early and middle directions; at $d=6n$, exact fixed-basis blindness reaches $2.5\%$ in the late strata of the $R_yR_z$--CZ and CNOT-line families, so we do not claim parameter-location independence.The complete sample-level CSV files, bootstrap outputs,
analysis notebooks, projection code, model-comparison code,
and figure-generation scripts are archived on
Zenodo~\cite{ait_haddou_2026_code_data}.  The crossover and aggregate experiments use seeds $20260804$ and $20260805$, respectively.

\clearpage
\twocolumngrid
\bibliography{readout_rank_references}

\end{document}